\documentclass[11pt]{article}
\usepackage{amsmath, amssymb, amsthm}
\usepackage{geometry}
\usepackage{cite}

\newtheorem{lemma}{Lemma}
\newtheorem{theorem}{Theorem}
\newtheorem{corollary}{Corollary}
\newtheorem{remark}{Remark}
\newtheorem{axiom}{Axiom}
\newtheorem{definition}{Definition}
\title{Homogeneity, Isotropy, and Determinism Force a Quadratic Spacetime Interval: A Derivation of Relativity Without Light}
\author{Deon Nicholas\\[4pt]
{\normalsize\itshape Independent Researcher, San Francisco, CA}\\[2pt]
{\normalsize\ttfamily deon.c.nicholas@gmail.com}}
\date{June 27, 2026}
\begin{document}
\maketitle
\begin{abstract}
We show that a few physical principles---smoothness, homogeneity, isotropy, and the determinism of inertial motion---force the invariant interval governing the geometry of spacetime to reduce to a quadratic form, without presupposing the existence of light or electromagnetic phenomena. Formalizing these as axioms about an ``invariant interval'' function $D:\mathbb{R}^n\to\mathbb{R}$ ($n\geq 3$), we find that smoothness and homogeneity force $D$ to be homogeneous of degree $p > 0$; determinism---that an inertial worldline be uniquely fixed by its initial point and direction---makes its geodesics straight lines; and isotropy---that the isometry group act transitively on each level set, with the stabilizer of a reference direction reversing every transverse direction---forces $D(v) = C\,(v^T S v)^{p/2}$ for a nondegenerate symmetric matrix $S$ and $p > 0$, with $p = 2$ (so that $D$ is exactly quadratic) when $S$ is indefinite. Thus the only admissible invariant intervals are powers of nondegenerate quadratic forms. The signature of $S$ is otherwise free: the definite case is Euclidean geometry and the indefinite case includes both Minkowski and ultrahyperbolic geometries, the two cases distinguished by the absence or presence of a null cone.

\end{abstract}

\section{Introduction}

Standard expositions of special relativity begin by postulating the constancy of the speed of light and derive the Lorentz transformations from this assumption. A well-established alternative---``relativity without light''---dispenses with the light postulate entirely. Beginning with von Ignatowsky \cite{Ignatowsky1910} and Frank and Rothe \cite{FrankRothe1911}, and developed further by Berzi and Gorini \cite{BerziGorini1969}, L\'evy-Leblond \cite{LevyLeblond1976}, and others \cite{Pal2003}, these derivations show that smoothness, homogeneity, and isotropy of space, together with the group structure of transformations between inertial frames, force the Lorentz transformations (or the Galilean transformations as a degenerate case). The speed of light appears not as an axiom but as an empirical parameter---or as a consequence.

These approaches, however, all begin with \emph{frames and velocities}: they assume inertial observers, coordinatize their relative motion, and derive the transformation law connecting them. Concepts like velocity, reference frames, and boosts are primitive.

In this paper, we start at a more primitive level. Our single primitive is an ``invariant interval'' function $D : \mathbb{R}^n \to \mathbb{R}$, capturing the invariant separation between events. The invariant interval is the central invariant of special relativity---the spacetime analogue of distance, whose value all inertial observers agree upon \cite{Misner1973}; but where it is ordinarily \emph{derived} from the Lorentz transformations, we take it as the primitive and ask what geometry it forces. We do not assume observers, frames, velocities, or any \emph{a priori} distinction between ``time'' and ``space.'' Instead, we impose a few structural principles on $D$---\emph{smoothness}, \emph{homogeneity}, \emph{isotropy}, and the \emph{determinism} of inertial motion---and show that these force $D$ to be a power of a nondegenerate quadratic form. The transformation group, the possibility of a null cone and an invariant speed, and the kinematic structure of special relativity are all consequences.

Our starting point is in several respects more primitive than in earlier light-free derivations. We assume neither a norm nor a metric: the interval $D$ may take either sign, need not obey a triangle inequality, and is not a distance function. Where those derivations begin with inertial frames and the group law relating their velocities, we begin only with $D$, a symmetry group encoding its isotropy, and a determinism principle; the concrete transformation laws---and, in the Lorentzian case, the null cone and causal structure---then emerge from the recovered geometry rather than being built in. From these principles we still recover the complete classification, across every nondegenerate signature---definite (Euclidean) and indefinite (Lorentzian and ultrahyperbolic) alike. The catalogue of geometries this yields is itself well-established---the Cayley--Klein kinematics and the possible-kinematics classification of Bacry and L\'evy-Leblond \cite{Yaglom1979,BacryLevyLeblond1968}; what is new here is the derivation of that catalogue from a single interval primitive, with the space--time split and the transformation group emerging rather than assumed.

We shall show that homogeneity (there is no preferred location or length scale) implies that $D$ depends only on displacement vectors and satisfies $D(\lambda v)=\lambda^p D(v)$ for some $p > 0$. From $D$ we construct a Lagrangian and define geodesics as stationary curves. A determinism principle---that the inertial worldline through a point is uniquely determined by its direction---forces geodesics to be straight lines. The isometries of $D$ are then shown to be affine, and isotropy, together with determinism, forces the level sets to be genuine (nondegenerate) ellipsoids or hyperboloids. The result is $D(v) = C\,(v^T S v)^{p/2}$ with $S$ a nondegenerate symmetric matrix; its signature is otherwise free, the definite case being Euclidean geometry and the indefinite case including Minkowski geometry, with the corresponding Lorentz transformations emerging as the symmetries of $D$.

\section{Mathematical Framework and Axioms}

\subsection{Basic Setup and Physical Motivation}

We model the universe as an $n$-dimensional vector space, which we identify with $\mathbb{R}^n$ (typically $n=4$, though we require only $n \geq 3$). We postulate that the fundamental physical principles of \emph{smoothness}, \emph{homogeneity}, \emph{isotropy}, and the \emph{determinism of inertial motion} hold; the symmetry content---of which isotropy is the strongest part---is our formalization of the \emph{principle of relativity}, that the laws of physics are the same in all inertial frames.

Central to our approach is the existence of invariant measures in the universe---quantities that are universal across all reference frames. Specifically, we postulate that there is a way to measure an invariant interval between two events that fundamentally represents a law of physics, independent of any observer. This interval encodes the geometric structure of the universe itself.

We therefore postulate the existence of a function $D$ that encodes this law of physics. The properties of $D$ are constrained by the physical principles above, which we now formalize as axioms.

\begin{axiom}[Smoothness]
\label{axiom:smoothness}
The interval function $D: \mathbb{R}^n \to \mathbb{R}$ is continuous on all of $\mathbb{R}^n$ and smooth ($C^\infty$) on $\mathbb{R}^n \setminus \{0\}$.
\end{axiom}

\textbf{Physical motivation:} Physical phenomena such as velocities, accelerations, and momenta are naturally expressed through derivatives. To build a coherent framework where such concepts are well-defined, we postulate that the laws of physics---and thus the interval function $D$---are smooth, with continuous derivatives of all orders. We ask for continuity everywhere, but only for differentiability away from the origin.

\begin{axiom}[Translational Homogeneity]
\label{axiom:translational}
The interval between two events $u, v \in \mathbb{R}^n$ depends only on their difference, not on their absolute positions:
\[
d(u,v) := D(u-v).
\]
\end{axiom}

\textbf{Physical motivation:} The laws of physics are the same everywhere. There is no preferred origin or distinguished event. This invariance under translation means intervals are determined solely by the displacement vector between events.

Without loss of generality, we study the single-argument function $D(v)$, understanding that the interval between events separated by displacement $v$ is $D(v)$. In particular $D(0) = d(x,x) = 0$: the interval between an event and itself, separated by zero displacement, vanishes. We note that $D$ may take any sign---positive, negative, or zero.

\begin{axiom}[Scale Homogeneity]
\label{axiom:scale}
For every $\lambda > 0$ and every $v \in \mathbb{R}^n$ with $D(v) \neq 0$,
\[
\frac{D(\lambda v)}{D(v)} = f(\lambda)
\]
for some function $f: (0,\infty) \to \mathbb{R}$.
\end{axiom}

\textbf{Physical motivation:} There is no preferred length scale built into the interval. Consider two arbitrary intervals, represented by displacement vectors $x$ and $y$ with $D(x), D(y) \neq 0$. The ratio $D(x)/D(y)$ captures a geometric relationship between these intervals. Now suppose we scale both intervals by the same factor $\lambda > 0$, obtaining $\lambda x$ and $\lambda y$. Since there is no fundamental scale, the relationship between the scaled intervals should be the same as between the original intervals:
\[
\frac{D(\lambda x)}{D(\lambda y)} = \frac{D(x)}{D(y)}.
\]
Re-arranging yields
\[
\frac{D(\lambda x)}{D(x)} = \frac{D(\lambda y)}{D(y)}.
\]
Since the left-hand side depends only on $x$ while the right-hand side depends only on $y$, and this must hold for arbitrary $x$ and $y$ (with $D(x), D(y) \neq 0$), both sides must equal a constant $f(\lambda)$ depending only on $\lambda$.

\begin{lemma}
\label{lem:power-law}
There exists a real exponent $p > 0$ such that $D(\lambda v) = \lambda^p D(v)$ for all $\lambda > 0$ and all $v$ with $D(v) \neq 0$. By continuity of $D$, this extends to $D(\lambda v) = \lambda^p D(v)$ for all $\lambda > 0$ and all $v \neq 0$.
\end{lemma}
\begin{proof}
From Axiom~\ref{axiom:scale}, we have $D(\lambda v) = f(\lambda) D(v)$ for all $\lambda > 0$ and $v$ with $D(v) \neq 0$.

Since $D$ is smooth (Axiom~\ref{axiom:smoothness}), the function $f(\lambda) = D(\lambda v)/D(v)$ for any fixed $v$ with $D(v) \neq 0$ is continuous in $\lambda$.

Consider the functional equation satisfied by $f$. For any $\lambda, \mu > 0$ and any fixed $v$ with $D(v) \neq 0$:
\begin{align*}
D(\lambda \mu v) &= f(\lambda \mu) D(v) \\
D(\lambda \mu v) &= f(\lambda) D(\mu v) = f(\lambda) f(\mu) D(v).
\end{align*}
Therefore, $f(\lambda \mu) = f(\lambda) f(\mu)$ for all $\lambda, \mu > 0$.

This is Cauchy's functional equation on $(0, \infty)$. Any continuous solution has the form $f(\lambda) = \lambda^p$ for some real constant $p$.

To see this, first note $f > 0$. It is nowhere zero, since $f(\lambda_0) = 0$ would give $f(\lambda) = f(\lambda_0)\,f(\lambda/\lambda_0) = 0$ for all $\lambda$, contradicting $f(1) = 1$; and a continuous, nowhere-vanishing function on the connected interval $(0, \infty)$ with $f(1) = 1 > 0$ is positive throughout. Now, taking logarithms, let $g(x) = \log f(e^x)$. Then
\begin{align*}
g(x+y)
&= \log f(e^{x+y}) \\
&= \log f(e^x e^y) \\
&= \log (f(e^x) f(e^y)) \\
&= \log f(e^x) + \log f(e^y) \\
&= g(x) + g(y)
\end{align*}
for all $x, y \in \mathbb{R}$, and $g$ is continuous (since $f$ is continuous and positive).

The additive Cauchy equation $g(x+y) = g(x) + g(y)$ with continuity forces $g$ to be linear: for any rational $r = m/n$ (with $m, n$ integers, $n > 0$), we have $ng(r) = g(nr) = mg(1)$, so $g(r) = rg(1)$. By continuity and density of rationals in $\mathbb{R}$, we conclude $g(x) = cx$ for all $x \in \mathbb{R}$, where $c = g(1)$. Setting $p = c$, we have $\log f(e^x) = px$, which gives $f(e^x) = e^{px}$, and substituting $\lambda = e^x$ yields $f(\lambda) = \lambda^p$.

We require $p > 0$: if $p \leq 0$ then for any $v$ with $D(v) \neq 0$ the value $D(\lambda v) = \lambda^p D(v)$ either diverges (if $p < 0$) or stays constant at $D(v) \neq 0$ (if $p = 0$) as $\lambda \to 0^+$, whereas continuity at the origin forces $D(\lambda v) \to D(0) = 0$. This contradiction rules out $p \leq 0$.

Therefore, $D(\lambda v) = \lambda^p D(v)$ for all $\lambda > 0$ and $v$ with $D(v) \neq 0$. For vectors $v$ with $D(v) = 0$ the identity still holds, both sides being zero: if instead $D(\lambda v) \neq 0$ for some $\lambda > 0$, then Axiom~\ref{axiom:scale} applied to the vector $\lambda v$ with scale factor $1/\lambda$ gives $D(v) = f(1/\lambda)\,D(\lambda v)$ with $f(1/\lambda) = (1/\lambda)^p \neq 0$, forcing $D(v) \neq 0$---a contradiction. Hence $D(\lambda v) = 0 = \lambda^p D(v)$.
\end{proof}

We say that $D$ is a \emph{homogeneous function of degree $p$}.

\begin{axiom}[Evenness]
\label{axiom:isotropy}
The interval function is invariant under reversal of direction. For every $v \in \mathbb{R}^n$,
\[
D(v) = D(-v).
\]
\end{axiom}

\textbf{Physical motivation:} Reversing the displacement between two events does not change their interval. A much stronger isotropy condition will appear in Axiom~\ref{axiom:isometry}.

\begin{lemma}[Extension to negative scalars]
\label{lem:absolute-homogeneity}
For all $\mu \neq 0$ and all $v \neq 0$,
\[
D(\mu v) = |\mu|^p D(v).
\]
In particular, $D(0) = 0$: by continuity at the origin (Axiom~\ref{axiom:smoothness}), letting $\mu \to 0$ gives $D(0) = \lim_{\mu \to 0} |\mu|^p D(v) = 0$, since $p > 0$.
\end{lemma}
\begin{proof}
For $\mu > 0$, this is Lemma~\ref{lem:power-law}. For $\mu < 0$, write $\mu = -|\mu|$ and compute:
\[
D(\mu v) = D(-|\mu| v) = D(|\mu| v) = |\mu|^p D(v),
\]
where the second equality uses Axiom~\ref{axiom:isotropy} and the third uses Lemma~\ref{lem:power-law} with $\lambda = |\mu| > 0$.
\end{proof}

\section{Geodesics and Inertial Motion}

We have established that $D$ is an even, homogeneous function of degree $p > 0$ encoding the invariant interval structure of the universe. We now show how $D$ naturally gives rise to the notion of \emph{geodesics}---the fundamental paths representing inertial motion.

\subsection{The Variational Principle}

If $D$ encodes a fundamental measure of intervals, then integrating $D$ along a curve accumulates the total ``interval'' traversed by that path. The paths of freely moving particles---those not subject to any external influence---should extremize this accumulated action.

This principle has deep roots in physics: from Fermat's principle in optics to the principle of least action in classical mechanics \cite{Misner1973}, the fundamental laws of nature describe motion through variational principles.

For a curve $\gamma: [a,b] \to \mathbb{R}^n$ parameterized by $t$, we define the action functional
\[
S[\gamma] = \int_a^b L(\dot{\gamma}(t)) \, dt,
\]
where $L(v) = D(v)$ is the \emph{Lagrangian}. Translation invariance (Axiom~\ref{axiom:translational}) ensures $L$ depends only on velocity $\dot{\gamma}(t)$, not position.

The Euler-Lagrange equations for such a Lagrangian are
\[
\frac{d}{dt} \left( \frac{\partial L}{\partial \dot{\gamma}} \right) = 0.
\]

To be precise about what ``extremize'' means: we say $S$ is \emph{stationary} at a curve $\gamma$ if its first-order change vanishes under every smooth perturbation that holds the endpoints fixed. That is, for every smooth $\eta:[a,b]\to\mathbb{R}^n$ with $\eta(a)=\eta(b)=0$,
\[
\left.\frac{d}{d\varepsilon}\right|_{\varepsilon=0} S[\gamma + \varepsilon\eta] = 0.
\]
A standard computation (differentiating under the integral sign and integrating by parts, using $\eta(a)=\eta(b)=0$) shows that $S$ is stationary at $\gamma$ if and only if $\gamma$ satisfies the Euler-Lagrange equations above. We take these stationary curves as our geodesics.

\subsection{Constructing the Lagrangian from $D$}

Given our interval function $D$, we set $L(v) = D(v)$. Since $D$ is homogeneous of degree $p$:
\[
L(\lambda v) = \lambda^p L(v).
\]

\begin{definition}[Geodesics]
\label{def:geodesic}
A \emph{geodesic} is a curve $\gamma(t)$ that satisfies the Euler-Lagrange equations for the Lagrangian $L(v) = D(v)$:
\[
\frac{d}{dt} \left( \frac{\partial L}{\partial \dot{\gamma}} \right) = 0.
\]
Since the gradient $\nabla D$ is defined only on $\mathbb{R}^n \setminus \{0\}$ (Axiom~\ref{axiom:smoothness}), this equation is meaningful only along curves with nowhere-vanishing velocity, $\dot\gamma(t) \neq 0$; these are the worldlines of genuinely moving particles, and they are the only curves we call geodesics.
\end{definition}

\begin{remark}
It is common in the literature to use $L(v) = D(v)^{1/p}$ instead, which is homogeneous of degree 1. On any cone where $D > 0$, the degree-one action $\int D(\dot\gamma)^{1/p}\,dt$ is invariant under orientation-preserving reparameterization, and the two Lagrangians $D$ and $D^{1/p}$ then determine the same unparametrized geodesics---a standard fact \cite[\S1.2--1.3]{BaoChernShen2000}---so nothing is lost by working with the more natural $L(v) = D(v)$. We do not invoke this in the parts of the argument concerning the null cone or the indefinite case, where $D$ may vanish or change sign and $D^{1/p}$ need not be smooth.
\end{remark}

\begin{axiom}[Determinism of Inertial Motion]
\label{axiom:uniqueness}
For any point $x_0 \in \mathbb{R}^n$ and any nonzero direction $v_0 \in \mathbb{R}^n \setminus \{0\}$, there is a unique inertial \emph{worldline} through $x_0$ in the direction $v_0$: every geodesic $\gamma$ with $\gamma(0) = x_0$ and $\dot\gamma(0)$ a positive multiple of $v_0$ traces out the same image.
\end{axiom}

\textbf{Physical motivation:} This is a principle of \emph{determinism}: the future of a free body is fixed by its present state. Given where an inertial particle is and which way it is heading, its entire worldline is determined---there is exactly one path it can trace, and it does trace it. This is among the most basic commitments of classical physics, and it is what makes inertial motion predictive rather than merely descriptive. We take it as a primitive axiom. We require $v_0 \neq 0$, as a worldline carries a genuine, nonzero direction of travel; this is also where the gradient $\nabla D(v_0)$ is defined.

Determinism forces $D \not\equiv 0$: were $D$ identically zero, the Lagrangian $L = D$ would be constant, every curve would be stationary, and no worldline through $x_0$ would be singled out---contradicting the axiom.

\begin{lemma}
\label{lem:straight-lines}
The worldline of every geodesic of the Lagrangian $L(v) = D(v)$ is a straight line.
\end{lemma}

\begin{proof}
Since $L = D$ depends only on velocity, the Euler--Lagrange equation reads
\[
\frac{d}{dt}\,\nabla D(\dot\gamma(t)) = 0,
\]
so a curve is a geodesic precisely when its \emph{momentum} $p(t) := \nabla D(\dot\gamma(t))$ is constant in time.\footnote{This is the simplest instance of Noether's theorem~\cite{Noether1918}: because $L=D$ depends only on velocity and not on position (translational homogeneity, Axiom~\ref{axiom:translational}), the associated momentum is conserved. More generally, every one-parameter family of isometries in $G$ yields a quantity conserved along geodesics.}

Fix $x_0 \in \mathbb{R}^n$ and a nonzero $v_0$. The line $\sigma(t) := x_0 + v_0 t$ has $\dot\sigma \equiv v_0 \neq 0$, so its momentum $\nabla D(v_0)$ is constant; hence $\sigma$ is a geodesic through $x_0$ in the direction $v_0$. By Axiom~\ref{axiom:uniqueness} every geodesic through $x_0$ in the direction $v_0$ traces the same image, namely this line. Therefore the worldline of every geodesic is a straight line.
\end{proof}

\subsection{The Isometry Group}

Having established that geodesics are straight lines, we now impose the final axiom encoding the principle of relativity: isotropy, realized by a rich symmetry group acting on the universe.

\begin{axiom}[Isotropy]
\label{axiom:isometry}
There exists a Lie group $G$ acting smoothly on $\mathbb{R}^n$ by diffeomorphisms, satisfying:
\begin{enumerate}
\item \textbf{Isometry.} Every $T \in G$ preserves intervals: $D(Tu - Tv)=D(u-v)$ for all $u, v \in \mathbb{R}^n$.

\item \textbf{Homogeneity of level sets.} For any nonzero $u, v \in \mathbb{R}^n$ with $D(u)=D(v)$, there exists an origin-fixing $T \in G$ such that $T(u) = v$.

\item \textbf{Reversibility.} Fix any $e_0$ with $D(e_0)\neq 0$, and let
\[
H := \{h \in G : h(0) = 0,\ h e_0 = e_0\}, \qquad
W := \{w \in \mathbb{R}^n : \nabla D(e_0) \cdot w = 0\}
\]
be the isometries fixing the origin and $e_0$, and the directions transverse to $e_0$ (tangent to the level set $\{D = D(e_0)\}$). For each nonzero $u \in W$, there exists $h \in H$ whose differential at $e_0$ satisfies $dh_{e_0}(u) = -u$.
\end{enumerate}
\end{axiom}

\textbf{Physical motivation.} Having established that geodesics represent inertial motion, we invoke the principle of relativity: the laws of physics are the same in all inertial frames, and the isometry group is the intrinsic freedom to change between them. Condition (2) is full rotational symmetry: directions of equal interval are physically equivalent, making $G$ transitive on each level set. Condition (3) is a \emph{reflection symmetry}: at a reference $e_0$, every transverse direction can be reversed by an isometry, so it is indistinguishable from its opposite. Together they make $G$ richly homogeneous.

\section{Structure of the Isometry Group}

We now use the isometry group to determine the precise form of $D$.

\subsection{Isometries are Affine Transformations}

\begin{lemma}[Fundamental Theorem of Affine Geometry]
\label{lem:lines-to-affine}
Let $T: \mathbb{R}^n \to \mathbb{R}^n$ be a continuous bijection such that both $T$ and $T^{-1}$ map straight lines into straight lines---equivalently, $T$ carries each line \emph{onto} a line. Then $T$ is an affine transformation: $T(x) = Ax + c$ for some invertible linear map $A$ and vector $c$.
\end{lemma}

\noindent The proof is given in Appendix~\ref{app:affine-geometry}.

\begin{lemma}[Isometries preserve geodesics]
\label{lem:isometries-preserve-geodesics}
Every $T \in G$ maps geodesics to geodesics.
\end{lemma}

\begin{proof}
Let $\gamma:[a,b]\to\mathbb{R}^n$ be any smooth curve and $T \in G$.

\medskip
\noindent\textbf{Step 1: $T$ preserves the interval of the velocity at every instant.}

We claim that
\begin{equation}
D\big(\dot{(T\gamma)}(t)\big) = D\big(\dot\gamma(t)\big) \qquad \text{for every } t,
\label{eq:speed-preserved}
\end{equation}
where $T\gamma := T \circ \gamma$. The proof uses only the definition of velocity as a difference quotient together with homogeneity.

For \emph{any} smooth curve $\sigma$, the velocity at time $t$ is the limit
\[
\dot\sigma(t) = \lim_{h \to 0^+} \frac{\sigma(t+h) - \sigma(t)}{h}.
\]
Since $D$ is continuous and homogeneous of degree $p$ (so $D(w/h) = D(w)/h^p$ for $h > 0$, by Lemma~\ref{lem:power-law}),
\[
D\big(\dot\sigma(t)\big)
= D\!\left(\lim_{h \to 0^+} \frac{\sigma(t+h) - \sigma(t)}{h}\right)
= \lim_{h \to 0^+} \frac{D\big(\sigma(t+h) - \sigma(t)\big)}{h^{p}}.
\]
Apply this identity twice: once to $\sigma = T\gamma$ and once to $\sigma = \gamma$. For each fixed $h$, the two numerators are equal, because the isometry condition (Axiom~\ref{axiom:isometry}(1)) applied to the points $u = \gamma(t+h)$ and $v = \gamma(t)$ gives
\[
D\big(T\gamma(t+h) - T\gamma(t)\big) = D\big(\gamma(t+h) - \gamma(t)\big).
\]
Letting $h \to 0^+$ on both sides yields \eqref{eq:speed-preserved}.

\medskip
\noindent\textbf{Step 2: $T$ preserves the action, hence geodesics.}

Recall the action functional $S[\gamma] = \int_a^b D(\dot\gamma(t))\,dt$---the total interval accumulated along the path---whose stationary curves are, by Definition~\ref{def:geodesic}, the geodesics. By \eqref{eq:speed-preserved} the integrand is unchanged at every instant, so
\[
S[T\gamma] = \int_a^b D\big(\dot{(T\gamma)}(t)\big)\,dt = \int_a^b D\big(\dot\gamma(t)\big)\,dt = S[\gamma]
\]
for \emph{every} smooth curve $\gamma$. Finally, $T$ maps stationary curves to stationary curves: $T$ is an invertible map carrying endpoints to endpoints, so if $\gamma$ is stationary and $\rho$ is any nearby curve with the same endpoints as $T\gamma$, then $T^{-1}\rho$ is a nearby curve with the same endpoints as $\gamma$ and $S[\rho] = S[T^{-1}\rho]$; since $S$ has no first-order change at $\gamma$, it has none at $T\gamma$ either. Hence $T$ maps geodesics to geodesics.
\end{proof}

\begin{theorem}
\label{thm:isometries-affine}
Every transformation $T \in G$ from Axiom~\ref{axiom:isometry} is an affine transformation.
\end{theorem}

\begin{proof}
By Lemma~\ref{lem:isometries-preserve-geodesics}, $T$ maps geodesics to geodesics. By Lemma~\ref{lem:straight-lines}, the worldline of every geodesic is a straight line. Therefore $T$ maps straight lines to straight lines. Since $T$ belongs to the group $G$ of smooth transformations (Axiom~\ref{axiom:isometry}), it has a smooth inverse $T^{-1} \in G$, which by the same argument also maps straight lines to straight lines; hence $T$ is a continuous bijection carrying lines \emph{onto} lines. Lemma~\ref{lem:lines-to-affine} then gives $T(x) = Ax + c$.
\end{proof}

\subsection{Linear Parts Preserve $D$}

\begin{lemma}
\label{lem:linear-part-preserves-D}
If $T(x) = Ax + c$ is an element of $G$, then $D(Av) = D(v)$ for all $v \in \mathbb{R}^n$.
\end{lemma}

\begin{proof}
By Axiom~\ref{axiom:isometry}(1):
\[
D(Tu - Tv) = D(A(u-v)) = D(u - v)
\]
for all $u, v$. Setting $w = u - v$ gives $D(Aw) = D(w)$ for all $w$.
\end{proof}

We denote the origin-fixing linear isometry group by
\[
G_0 = \{A \in \mathrm{GL}(\mathbb{R}^n) : D(Av) = D(v) \text{ for all } v\}.
\]

\section{The Interval Function is Determined by a Quadratic Form}

We now prove the central result: under our axioms, $D$ must be a power of a nondegenerate quadratic form.

\subsection{Sign regions and cones}

Recall that $D$ is smooth away from the origin, even, and homogeneous of degree $p>0$:
\[
D(\lambda v) = |\lambda|^p D(v) \quad \text{for all } \lambda \neq 0,\, v\neq 0.
\]
Define the sign regions
\[
\Omega_+ := \{v \in \mathbb{R}^n \setminus \{0\} : D(v) > 0\}, \qquad
\Omega_- := \{v \in \mathbb{R}^n \setminus \{0\} : D(v) < 0\}.
\]

\begin{lemma}[Cones from homogeneity]
\label{lem:cones}
Each connected component of $\Omega_+$ and $\Omega_-$ is a cone: if $v$ lies in a given component and $\lambda>0$, then $\lambda v$ lies in the same component.
\end{lemma}

\begin{proof}
Let $C$ be a connected component of $\Omega_+$. If $v\in C$, then $D(v)>0$. For any $\lambda>0$, $D(\lambda v) = \lambda^p D(v) > 0$, so $\lambda v \in \Omega_+$. The map $\lambda \mapsto \lambda v$ is continuous and connects $v$ to any point on the ray, so all such points lie in the same connected component. The argument for $\Omega_-$ is identical.
\end{proof}

Since $D \not\equiv 0$, at least one of $\Omega_+, \Omega_-$ is nonempty. We develop the theory under the hypothesis $\Omega_+ \neq \emptyset$, which loses no generality: the isotropy axiom is anchored at any reference direction of \emph{nonzero} interval, of either sign, so when $\Omega_+ \neq \emptyset$ we may take the reference $e_0 \in \Omega_+$ (normalized to $D(e_0) = 1$) and build the geometry there. The complementary case $\Omega_+ = \emptyset$---where $D \leq 0$ everywhere---is treated separately, by passing to $-D$, in Lemma~\ref{lem:negative-case}.

We fix a connected component $C$ of $\Omega_+$ and work within this cone.

\begin{lemma}[One point per ray]
\label{lem:one-point-per-ray}
Let $C$ be a connected component of $\Omega_+$. The level set
$S := \{v \in C : D(v) = 1\}$
intersects each ray through the origin in exactly one point.
\end{lemma}

\begin{proof}
For $v \in C$, the equation $D(\lambda v) = \lambda^p D(v) = 1$ has unique solution $\lambda = D(v)^{-1/p}$.
\end{proof}

\subsection{Normalized interval function}

To simplify the analysis, we reduce to the degree-2 case.

\begin{definition}[Normalized interval function]
\label{def:Q-normalized}
On each connected component of $\{D > 0\}$, define
\[
Q(v) := D(v)^{2/p}.
\]
\end{definition}

\begin{lemma}[Properties of $Q$]
\label{lem:Q-properties}
The function $Q$ satisfies:
\begin{enumerate}
\item $Q$ is smooth on $C$ (and on each connected component of $\{D > 0\}$).
\item $Q$ is homogeneous of degree 2: $Q(\lambda v) = \lambda^2 Q(v)$ for all $\lambda > 0$.
\item $G_0$ preserves $Q$: $Q(Av) = Q(v)$ for all $A \in G_0$.
\item The level sets of $D$ and $Q$ coincide: $\{D = 1\} = \{Q = 1\}$ on $C$.
\end{enumerate}
\end{lemma}

\begin{proof}
Properties (1), (3), (4) follow from the corresponding properties of $D$, using that $D > 0$ on $C$ so $D^{2/p}$ is well-defined and smooth. For (2):
$Q(\lambda v) = D(\lambda v)^{2/p} = (\lambda^p D(v))^{2/p} = \lambda^2 D(v)^{2/p} = \lambda^2 Q(v)$.
\end{proof}

\begin{lemma}[Regular level set]
\label{lem:regular-level}
The level set $M := \{v \in C : Q(v) = 1\}$ is a smooth embedded hypersurface in $\mathbb{R}^n$.
\end{lemma}

\begin{proof}
Since $Q$ is 2-homogeneous, Euler's theorem for homogeneous functions gives
\[
\nabla Q(v) \cdot v = 2Q(v)
\]
for all $v \in C$. At any $v \in M$, we have $Q(v) = 1$, so $\nabla Q(v) \cdot v = 2 \neq 0$. In particular $\nabla Q(v) \neq 0$ for every $v \in M$. By the implicit function theorem, $M$ is a smooth hypersurface.
\end{proof}

\begin{lemma}[Transitivity on the level set]
\label{lem:preserve-cone}
$G_0$ preserves $Q$ throughout $\Omega_+$, and acts transitively on $M = \{Q = 1\} \cap C$.
\end{lemma}

\begin{proof}
Preservation of $Q$ is Lemma~\ref{lem:Q-properties}(3). For transitivity, let $u, v \in M$. Then $D(u) = D(v) = 1$, so Axiom~\ref{axiom:isometry}(2) provides an origin-fixing isometry $A \in G_0$ with $Au = v$. Since $u \in C$ and $Au = v \in C$, this $A$ carries the component $C$ to itself; hence $G_0$ is transitive on $M$. (Individual elements of $G_0$ need not fix $C$: in the indefinite case $-I \in G_0$ exchanges the two cones of $\Omega_+$.)
\end{proof}

\subsection{Proof that $Q$ is a quadratic form}

\begin{theorem}[Q is a quadratic form]
\label{thm:Q-quadratic}
Under Axioms~\ref{axiom:smoothness}--\ref{axiom:isometry}, the normalized interval function $Q$ has the form
\[
Q(x) = x^T S x
\]
for some fixed symmetric matrix $S \in \mathbb{R}^{n \times n}$, constant throughout $C$.
\end{theorem}

\textbf{Strategy.} We work on the level set $M = \{Q = 1\} \subset C$, fix a basepoint $e_0 \in M$, and use the stabilizer $H = \{h \in G_0 : he_0 = e_0\}$. The key steps are: (1) decompose $\mathbb{R}^n$ into the reference direction and the tangent space $W$; (2) use the central symmetry $u\mapsto -u$ on $W$ supplied by isotropy (Axiom~\ref{axiom:isometry}(3)) to show that every $H$-invariant symmetric trilinear form on $W$ vanishes; (3) deduce that the third derivative of $Q$ vanishes; (4) propagate to all of $C$ and conclude.

\begin{proof}
Throughout, we work on a fixed connected component $C$ of $\Omega_+$ and the level set $M = \{Q = 1\} \subset C$.

\medskip
\noindent\textbf{Step 1: Tangent space decomposition.}

Since $Q$ is 2-homogeneous, Euler's theorem gives $\nabla Q(x) \cdot x = 2Q(x)$. At $e_0 \in M$:
\[
\nabla Q(e_0) \cdot e_0 = 2 \neq 0,
\]
so $\nabla Q(e_0) \neq 0$. The tangent space to $M$ at $e_0$ is
\[
W := T_{e_0} M = \{w \in \mathbb{R}^n : \nabla Q(e_0) \cdot w = 0\}.
\]
Since $\nabla Q(e_0) \cdot e_0 = 2 \neq 0$ (Euler's theorem), $e_0 \notin W$. As $W$ has codimension 1:
\begin{equation}
\mathbb{R}^n = \mathbb{R} e_0 \oplus W.
\tag{$\star$}
\end{equation}

\medskip
\noindent\textbf{Step 2: A central symmetry on $W$.}

The elements of $H \subseteq G_0$ are origin-fixing isometries, hence linear (Theorem~\ref{thm:isometries-affine}), so each $h$ coincides with its differential $dh_{e_0}$ and preserves $D$---and therefore $Q = D^{2/p}$. Differentiating $Q(hx)=Q(x)$ at $x = e_0$ gives $\nabla Q(e_0)\,h = \nabla Q(e_0)$, so each $h$ preserves $W = \ker \nabla Q(e_0)$; since $\nabla Q(e_0) = \tfrac{2}{p}\,D(e_0)^{2/p - 1}\,\nabla D(e_0)$ is a positive multiple of $\nabla D(e_0)$, this is the same subspace whether $W$ is read through $Q$ or through $D$.

Axiom~\ref{axiom:isometry}(3) supplies, for every nonzero $u \in W$, an element $h \in H$ with $dh_{e_0}(u) = -u$; since each $h \in H$ is linear, $dh_{e_0} = h$ and hence $hu = -u$.

\begin{corollary}
\label{cor:H-invariant-forms}
Every $H$-invariant symmetric trilinear form on $W$ is identically zero.
\end{corollary}

\begin{proof}
Let $T$ be such a form and set $f(u) := T(u,u,u)$; then $f$ is $H$-invariant and odd, $f(-u) = -f(u)$. Fix a nonzero $u \in W$. By the remark above there is $h \in H$ with $hu = -u$, whence
\[
f(u) = f(hu) = f(-u) = -f(u),
\]
so $f(u) = 0$. Hence $T(u,u,u) = 0$ for all $u \in W$, and by polarization $T \equiv 0$.
\end{proof}

\medskip
\noindent\textbf{Step 3: The third derivative of $Q$ vanishes at $e_0$.}

Let $T: \mathbb{R}^n \times \mathbb{R}^n \times \mathbb{R}^n \to \mathbb{R}$ denote the third derivative tensor of $Q$ at $e_0$:
\[
T(u, v, w) := D^3 Q(e_0)[u, v, w].
\]

Since $Q(hx) = Q(x)$ for all $h \in H$, differentiating three times at $x = e_0$ gives:
\begin{equation}
T(hu, hv, hw) = T(u, v, w) \quad \text{for all } h \in H.
\label{eq:T-H-invariant}
\end{equation}

\textbf{Step 3.1: $T|_{W \times W \times W} = 0$.}

The restriction of $T$ to $W^3$ is an $H$-invariant symmetric trilinear form on $W$. By Corollary~\ref{cor:H-invariant-forms}:
\begin{equation}
T(u, v, w) = 0 \quad \text{for all } u, v, w \in W.
\label{eq:T-WWW}
\end{equation}

\textbf{Step 3.2: $T$ vanishes whenever $e_0$ is one of the arguments.}

Since $Q$ is $2$-homogeneous, $\mathrm{Hess}\,Q$ is homogeneous of degree $0$, i.e.\ constant along rays:
\[
\mathrm{Hess}\,Q(\lambda e_0) = \mathrm{Hess}\,Q(e_0) \qquad (\lambda > 0).
\]
Differentiating in $\lambda$ at $\lambda = 1$ gives $D^3 Q(e_0)[e_0, \cdot, \cdot] = 0$ (Euler's identity for the degree-$0$ tensor $\mathrm{Hess}\,Q$). By symmetry of $T$, the same holds with $e_0$ in any slot:
\begin{equation}
T(e_0, y, z) = 0 \quad \text{for all } y, z \in \mathbb{R}^n.
\label{eq:T-e0}
\end{equation}
In particular this covers $T(e_0, u, v)$, $T(e_0, e_0, u)$, and $T(e_0, e_0, e_0)$.

\textbf{Step 3.3: Conclusion: $T \equiv 0$ on $\mathbb{R}^n$.}

By the decomposition $\mathbb{R}^n = \mathbb{R} e_0 \oplus W$, every vector is $a e_0 + w$. By trilinearity, $T(x,y,z)$ expands into terms each of which has all three arguments in $W$ (vanishing by (\ref{eq:T-WWW})) or has at least one argument equal to $e_0$ (vanishing by (\ref{eq:T-e0})). Hence:
\begin{equation}
D^3 Q(e_0) = T \equiv 0.
\label{eq:T-zero}
\end{equation}

\medskip
\noindent\textbf{Step 4: Propagate to all of $C$ and conclude.}

Because $G_0$ preserves $Q$, the third derivative transforms equivariantly:
\[
D^3 Q(gx)[u, v, w] = D^3 Q(x)[g^{-1}u, g^{-1}v, g^{-1}w].
\]
By transitivity of $G_0$ on $M$ (Lemma~\ref{lem:preserve-cone}), equation (\ref{eq:T-zero}) implies $D^3 Q(x) = 0$ for all $x \in M$.

By 2-homogeneity, $D^3 Q(\lambda x) = \lambda^{-1} D^3 Q(x)$, so $D^3 Q \equiv 0$ on the entire cone $C$.

Therefore $\mathrm{Hess}\, Q$ is constant on $C$: $\mathrm{Hess}\, Q(x) \equiv \mathcal{H}$ for some fixed symmetric matrix $\mathcal{H}$. Integrating twice:
\[
Q(x) = \tfrac{1}{2} x^T \mathcal{H} x + (\text{affine-linear terms}).
\]
But $Q$ is 2-homogeneous, forcing the affine-linear terms to vanish. Setting $S := \tfrac{1}{2}\mathcal{H}$:
\[
Q(x) = x^T S x \quad \text{for all } x \in C. \qedhere
\]
\end{proof}

\subsection{Forced nondegeneracy}

Theorem~\ref{thm:Q-quadratic} produces a symmetric matrix $S$ with $Q = x^T S x$ on $C$. Throughout the rest of this section we write
\[
q(x) := x^T S x,
\]
regarded as a single quadratic form defined on all of $\mathbb{R}^n$; thus $q = Q$ on the cone $C$, while $q$ extends to every $x \in \mathbb{R}^n$. Determinism now forces $S$ to be nondegenerate.

\begin{theorem}[Forced nondegeneracy]
\label{thm:forced-nondegeneracy}
The symmetric matrix $S$ with $Q(x) = x^T S x$ is nondegenerate.
\end{theorem}

\begin{proof}
Suppose not: $S w_0 = 0$ for some $w_0 \neq 0$. Fix $v_0 \in C$ and set $\gamma(t) := x_0 + v_0 t + \tfrac{1}{2} w_0 t^2$, so $\dot\gamma(t) = v_0 + t\,w_0$. Since $S w_0 = 0$,
\[
q(\dot\gamma(t)) = q(v_0) + 2t\,v_0^T S w_0 + t^2 q(w_0) = q(v_0) > 0,
\]
so $\dot\gamma(t)$ has constant positive interval $q(v_0)$. Because $C$ is open and $\dot\gamma(0) = v_0 \in C$, there is an interval $t \in (-\delta, \delta)$ on which $\dot\gamma(t) \in C$; there $D = q^{p/2}$, so
\[
\nabla D(\dot\gamma(t)) = p\,q(\dot\gamma(t))^{p/2 - 1}\, S\,\dot\gamma(t) = p\,q(v_0)^{p/2-1}\,S v_0
\]
is constant (the $S w_0$ term drops out). Thus $\gamma|_{(-\delta,\delta)}$ has constant momentum and is a geodesic with $\gamma(0)=x_0$, $\dot\gamma(0)=v_0$. Its image is a parabolic arc---genuinely curved, since $S v_0 \neq 0 = S w_0$ makes $w_0$ independent of $v_0$---not the straight line through $x_0$ in the direction $v_0$, contradicting Axiom~\ref{axiom:uniqueness}. Hence $S$ is nondegenerate.
\end{proof}

The signature of $S$ is left free: any nondegenerate symmetric form is admissible at this stage. The definite signatures give Euclidean geometry; the indefinite ones (including the Lorentzian signature $(1, n{-}1)$ and, when $n \geq 4$, ultrahyperbolic signatures such as $(2,2)$) are pinned down further in the next subsection, where determinism forces $p = 2$.

\subsection{Extension to all of $\mathbb{R}^n$}
\label{sec:extension}

Theorem~\ref{thm:Q-quadratic} determines $D$ only on the cone $C$, where $D = q^{p/2}$ with $q(x)=x^T S x$ and $S$ nondegenerate. We now extend this to all of $\mathbb{R}^n$. The definite and indefinite signatures behave differently, and only the indefinite one requires work.

\paragraph{Definite case.}
If $S$ is definite, then $q$ has constant sign, so $\{q > 0\} = \mathbb{R}^n \setminus \{0\}$ is connected and equals $C$. Theorem~\ref{thm:Q-quadratic} already gives $D = q^{p/2}$ on all of $\mathbb{R}^n \setminus \{0\}$, smooth for every $p > 0$. Nothing further is required.

\paragraph{Indefinite case.}
When $S$ is indefinite the cone $C$ is a proper open subset of $\mathbb{R}^n \setminus \{0\}$, and its boundary lies on the null cone of $S$. Smoothness of $D$ across that boundary forces $p = 2$, after which a single quadratic form is pinned down everywhere.

\begin{lemma}[Global form]
\label{lem:global-form}
If $S$ is indefinite, then $p = 2$ and $D(x) = x^T S x$ on all of $\mathbb{R}^n$.
\end{lemma}

\begin{proof}
\emph{The boundary of $C$ is null.} Since $S$ is indefinite, $q$ takes negative values, so $C \subseteq \{q > 0\}$ is a proper subset of $\mathbb{R}^n \setminus \{0\}$ and $\partial C \neq \emptyset$. A boundary point $x_0 \in \partial C$ is a limit of points of the open set $\{D > 0\}$ but does not belong to it, so $D(x_0) = 0$ by continuity; and on $C$, $D = q^{p/2} \to 0$ forces $q(x_0) = 0$. Thus $\partial C \subseteq \{q = 0\}$.

\emph{Smoothness forces $p = 2$.} Pick $x_0 \in \partial C$ with $\nabla q(x_0) = 2 S x_0 \neq 0$ (such points are dense in $\{q = 0\}$ as $S$ is nondegenerate). Then $s := q$ is one coordinate of a smooth chart at $x_0$, with $\{s > 0\}$ locally inside $C$, where $D = s^{p/2}$. For $D$ to be $C^\infty$ at $x_0$, the exponent $p/2$ must be a positive integer $m$ (otherwise $s^{p/2}$ fails to be $C^\infty$ at $s = 0$); so $D = q^m = (x^T S x)^m$ on $C$.

It remains to show $m = 1$; suppose $m \geq 2$. On the open cone $C$ the momentum of the Lagrangian $L = D = q^m$ is
\[
\nabla D(v) = m\,q(v)^{m-1}\,\nabla q(v) = 2m\,q(v)^{m-1}\,S v.
\]
Since $D$ is $C^1$ on $\mathbb{R}^n \setminus \{0\}$ (Axiom~\ref{axiom:smoothness}) and $q(v) \to 0$ as $v \to \partial C$ while $S v$ stays bounded, letting $v \to \nu$ through $C$ gives, because $m - 1 \geq 1$,
\[
\nabla D(\nu) = \lim_{v \to \nu,\, v \in C} 2m\,q(v)^{m-1}\,S v = 0 \qquad \text{for every } \nu \in \partial C.
\]
So the momentum vanishes on the whole null boundary $\partial C$. Pick a regular boundary point $\nu_0 \in \partial C$ with $S \nu_0 \neq 0$ (such points are dense, as above). Near $\nu_0$ the null cone $\{q = 0\}$ is a smooth hypersurface of dimension $n - 1 \geq 2$. We first check that $\partial C$ coincides with $\{q = 0\}$ near $\nu_0$. Choose a neighborhood $U$ of $\nu_0$ small enough that $V := U \cap \{q > 0\}$ is connected. Since $C$ accumulates at $\nu_0 \in \partial C$, the set $C \cap U = C \cap V$ is nonempty and open in $V$; it is also closed in $V$, for a point of $V$ approached by $C$ but not in $C$ would lie in $\partial C \subseteq \{q = 0\}$, impossible in $V$. By connectedness $C \cap U = V = U \cap \{q > 0\}$, so $\partial C \cap U = U \cap \{q = 0\}$ and the local boundary curve below indeed lies in $\partial C$. Its tangent space at $\nu_0$ is $(S\nu_0)^\perp$, which contains $\nu_0$ (as $q(\nu_0) = 0$) but, having dimension $\geq 2$, also a vector $w \notin \mathbb{R}\nu_0$. Take a smooth curve $\nu : (-\delta, \delta) \to \partial C$ with $\nu(0) = \nu_0$ and $\nu'(0) = w$. Then
\[
\gamma(t) = x_* + \int_0^t \nu(s)\,ds, \qquad \dot\gamma(t) = \nu(t),
\]
has null velocity throughout, so its momentum $\nabla D(\dot\gamma(t)) \equiv 0$ is constant and $\gamma$ is a geodesic. The straight line $\sigma(t) = x_* + t\,\nu_0$ shares the initial data $\sigma(0) = \gamma(0) = x_*$, $\dot\sigma(0) = \dot\gamma(0) = \nu_0$ and is also a geodesic (its constant velocity $\nu_0 \in \partial C$ has momentum $0$). Were $\gamma$ that same line, its velocity direction would be constant, $\nu(t) \in \mathbb{R}\nu_0$, forcing $\nu'(0) = w \in \mathbb{R}\nu_0$ and contradicting the choice of $w$. So $\gamma$ traces a curve, not a line, while sharing initial data with $\sigma$---contradicting Axiom~\ref{axiom:uniqueness}. Hence $m = 1$: $p = 2$ and $D = q = x^T S x$ on $C$.

\emph{The quadratic form spreads to all of $\mathbb{R}^n$.} With $p = 2$ we have $Q = D^{2/p} = D$, so Theorem~\ref{thm:Q-quadratic} gives $D^3 D \equiv 0$ on $C$, in particular on $C \cap \{D = 1\}$. Each origin-fixing $A \in G_0$ is linear with $D \circ A = D$, so the third derivative transforms equivariantly, $D^3 D(Ax)[Au, Av, Aw] = D^3 D(x)[u, v, w]$; hence $D^3 D$ vanishes at $Ax$ whenever it vanishes at $x$. By Axiom~\ref{axiom:isometry}(2), $G_0$ is transitive on $\{D = 1\}$, so $D^3 D \equiv 0$ there. The third-derivative tensor $D^3 D$ is homogeneous of degree $p - 3$ (each derivative lowers the degree of the degree-$p$ function $D$ by one), so $D^3 D(\lambda x) = \lambda^{p-3} D^3 D(x)$ for $\lambda > 0$; since this factor never vanishes, a zero of $D^3 D$ propagates along its entire ray. Every ray of the cone $\{D > 0\}$ meets $\{D = 1\}$, so $D^3 D \equiv 0$ on all of $\{D > 0\}$, and $\mathrm{Hess}\,D$ (its second derivative) is therefore constant there, equal to its value $2S$ on $C$.

This covers only the positive cone; the negative region requires a separate pass. First, $\{D < 0\}$ is nonempty. Were $D \geq 0$ everywhere, then a boundary point $x_0 \in \partial C$---which exists and is null by the first step---would be a global minimum of $D$, so its Hessian would be positive semidefinite, $\mathrm{Hess}\,D(x_0) \succeq 0$. But $\mathrm{Hess}\,D = 2S$ throughout $\{D > 0\}$, and continuity at $x_0 \in \partial C$ gives $\mathrm{Hess}\,D(x_0) = 2S$, which is \emph{indefinite} (having both positive and negative eigenvalues, since $S$ is). A semidefinite matrix cannot be indefinite---contradiction---so $\{D < 0\} \neq \emptyset$. Now $-D$ also satisfies all the axioms: it is smooth off the origin, even, homogeneous of degree $2$, has the same stationary curves (negating the action functional does not move its critical points) and hence the same determinism, and shares the isometry group and level sets of $D$. Its positive cone is $\{-D > 0\} = \{D < 0\}$, now known nonempty. Because $p = 2$ was established above, the degree-one normalization of $-D$ is just $(-D)^{2/p} = -D$ itself, so Theorem~\ref{thm:Q-quadratic} applies to $-D$ directly. Running it and the spreading argument of the previous paragraph with $-D$ in place of $D$ then yields a symmetric matrix $\tilde S$ with $\mathrm{Hess}(-D) \equiv 2\tilde S$, i.e.\ $\mathrm{Hess}\,D \equiv -2\tilde S$, constant on $\{D < 0\}$. At a null point $x_0$ with $2 S x_0 \neq 0$, continuity of the gradient gives $\nabla D(x_0) = \lim_{x \to x_0,\, x \in C} 2 S x = 2 S x_0 \neq 0$, so $D$ is regular there and both $\{D > 0\}$ and $\{D < 0\}$ accumulate at $x_0$; continuity of $\mathrm{Hess}\,D$ then forces $-2\tilde S = 2S$. Hence $\mathrm{Hess}\,D = 2S$ on $\{D > 0\}$, on $\{D < 0\}$, and---by continuity---on the null cone, i.e.\ on all of $\mathbb{R}^n \setminus \{0\}$. Integrating and using $2$-homogeneity, $D(x) = x^T S x$ everywhere.

The obstruction to $m \geq 2$ is structural: the momentum $2m\,q^{m-1} S v$ vanishes identically on the null cone $\{q = 0\}$, so the gradient cannot distinguish one null direction from another and fails to determine the worldline there. Only $m = 1$ keeps the momentum $2 S v$ nonzero across the cone, so that inertial motion stays determined.
\end{proof}

\subsection{Main result}
\label{sec:from-Q-to-D}

\begin{theorem}[Main result]
\label{thm:main}
Under Axioms~\ref{axiom:smoothness}--\ref{axiom:isometry}, the interval function has the form
\[
D(v) = C\,(v^T S v)^{p/2},
\]
where $S$ is a nondegenerate symmetric matrix, $C \neq 0$ is a constant, and $p > 0$. When $\Omega_+ \neq \emptyset$ one has $C > 0$; a negative constant $C < 0$ occurs only in the negative-definite case $\Omega_+ = \emptyset$ (Lemma~\ref{lem:negative-case}). If $S$ has indefinite signature, then $p = 2$, so that $D(v) = C\,v^T S v$ is exactly a quadratic form.
\end{theorem}

\begin{proof}
Theorem~\ref{thm:Q-quadratic} gives $D = (v^T S v)^{p/2}$ on the cone $C$, and Theorem~\ref{thm:forced-nondegeneracy} shows $S$ is nondegenerate. In the definite case $C = \mathbb{R}^n \setminus \{0\}$, so this is already the global form, valid for any $p > 0$; different values of $p$ yield the same unparameterized geodesic images and the same isometry group, so the degree is a matter of convention (one may take $p = 2$). In the indefinite case Lemma~\ref{lem:global-form} extends the formula to all of $\mathbb{R}^n$ and forces $p = 2$, so $D$ is exactly the quadratic form $C\,v^T S v$. The constant $C > 0$ accounts for the normalization $D(e_0) = 1$ used above; this proves the theorem when $\Omega_+ \neq \emptyset$. The remaining case $\Omega_+ = \emptyset$ is Lemma~\ref{lem:negative-case}, which gives the same form with $C < 0$. In every case $C \neq 0$.
\end{proof}

\begin{remark}[The isometry group need not be connected]
\label{rem:disconnected}
The group $G$ is not assumed connected. In the indefinite case the level set $\{D = 1\} = \{q = 1\}$ can be disconnected---for the Lorentzian signature $(1, n{-}1)$ it is the two-sheeted timelike hyperboloid---yet Axiom~\ref{axiom:isometry}(2) requires $G_0$ to act transitively on it. A discrete isometry must therefore connect the components; one is always present, namely $-I$, which lies in $G_0$ by evenness. For $D = C\,v^T S v$ the full origin-fixing isometry group is $O(S)$, in which the assumed $G_0$ embeds as a subgroup large enough to meet the transitivity conditions of Axiom~\ref{axiom:isometry}.
\end{remark}

\subsection{Negative interval functions}
\label{sec:negative-case}

The theorem was proved under the standing hypothesis $\Omega_+ = \{D > 0\} \neq \emptyset$. The complementary case $\Omega_+ = \emptyset$---where $D \leq 0$ on all of $\mathbb{R}^n \setminus \{0\}$---reduces to it at once, because the isotropy axiom is anchored at any direction of \emph{nonzero} interval, of either sign.

\begin{lemma}[Negative case]
\label{lem:negative-case}
If $\Omega_+ = \emptyset$ and $D \not\equiv 0$, then $D(v) = C\,(v^T S v)^{p/2}$ for a constant $C < 0$ and a positive-definite symmetric matrix $S$.
\end{lemma}

\begin{proof}
Apply the already-proved $\Omega_+ \neq \emptyset$ case to $D' := -D$. Then $\{D' > 0\} = \{D < 0\} \neq \emptyset$, and $D'$ satisfies every axiom. Smoothness and even homogeneity of degree $p$ are immediate; determinism holds because scaling a Lagrangian by $-1$ leaves its stationary curves unchanged. The isometry conditions transfer verbatim: $D'$ has the same origin-fixing isometry group $G_0$, the same transverse spaces $W = \ker\nabla D'(e_0) = \ker\nabla D(e_0)$, and the same equal-interval relation $\{D'(u)=D'(v)\} = \{D(u)=D(v)\}$, so conditions (1)--(3) read identically for $D'$ and $D$.

Since $\{D' > 0\}\neq\emptyset$, the $\Omega_+ \neq \emptyset$ case gives $D'(v) = C'\,(v^T S v)^{p/2}$ with $C' > 0$ and $S$ nondegenerate, so $D = -D' = C\,(v^T S v)^{p/2}$ with $C = -C' < 0$. Finally $S$ is definite: if $S$ were indefinite, that case would force $p = 2$ and $D' = C'\,v^T S v$, which then takes negative values, giving $\{D > 0\} = \{D' < 0\} \neq\emptyset$ and contradicting $\Omega_+ = \emptyset$. Hence $S$ is positive-definite.
\end{proof}

\section{Conclusion}

We have shown that any invariant interval function $D$ satisfying smoothness, homogeneity, isotropy, and the determinism of inertial motion must take the form
\[
D(v) = C\,(v^T S v)^k
\]
where $S$ is a nondegenerate symmetric matrix, $C \neq 0$, and $k = p/2$. In the indefinite case, $k = 1$, so $D$ is exactly a nondegenerate quadratic form; in the definite (Euclidean) case, any $k > 0$ yields a valid interval function (different exponents share the same unparameterized geodesic images and isometries, so $k$ is a matter of convention).

The signature of $S$ is not further constrained: every nondegenerate symmetric form is realized by some interval function obeying the axioms. The definite signatures give Euclidean geometry; among the indefinite ones, the Lorentzian signature $(1, n{-}1)$ and its reverse yield special relativity, while ultrahyperbolic signatures such as $(2,2)$ are equally admissible when $n \geq 4$. Singling out the Lorentzian case from the other indefinite ones requires an additional principle beyond symmetry and determinism---a notion of causality, or a single distinguished time direction---which our axioms deliberately do not impose. What the axioms \emph{do} force, however, is the quadratic structure and the nondegeneracy.

\medskip
\noindent\textbf{From intervals to physics.}
With the form of $D$ now determined, we can recover the familiar language of physics. Every universe satisfying our axioms has \emph{a form of relativity}: the isometry group of $D$ mixes directions nontrivially, so an observer's decomposition into ``the direction of motion'' and ``everything else'' depends on the observer. What differs with the signature is the \emph{character} of this relativity.

In the Minkowski case (signature $(1, n{-}1)$), the isometry group is $O(1, n{-}1)$, whose elements include Lorentz boosts---hyperbolic rotations that mix an observer's direction of motion with the transverse directions. The zero set $\{D = 0\}$ defines null directions and an invariant speed: if particles travel along null directions, their speed is a property of the geometry, not the particle. This is Einstein's relativity in its familiar form: time dilation, length contraction, and the relativity of simultaneity all governed by hyperbolic geometry.

The Euclidean case admits a more interpretive reading. Here the analogy with relativity is formal rather than physical---there are no null directions, no causal structure, and no maximum speed---but the symmetry-theoretic skeleton is the same, and it is worth drawing out. In signature $(n,0)$ the isometry group is $O(n)$ and changes of observer are ordinary rotations. If one nonetheless designates a direction of travel as an observer's ``time'' axis, with the orthogonal directions as that observer's ``space,'' then durations, lengths, and the simultaneity of distinct events are all read off from the single interval function $D$---now the Euclidean (Pythagorean) form rather than the hyperbolic one. Because $O(n)$ acts transitively on directions, no such axis is privileged: a rotation carries one observer's choice to another's, mixing ``time'' and ``space,'' so different observers disagree on the decomposition and on the durations and lengths they assign---governed now by spherical rather than hyperbolic geometry. What persists from the Lorentzian case is thus not the causal content of relativity but its core structural feature: an observer-dependent split of spacetime, enforced by a transitive symmetry group.

Galilean physics---in which all observers agree on the passage of time---does not appear. It would require a degenerate interval function ($\det S = 0$); but isotropy already forces $D$ to be built from a quadratic form, and determinism then forces that form to be nondegenerate---a degenerate $S$ leaves inertial worldlines undetermined, exactly the pathology ruled out above. Galilean relativity is thus not a third option but a degenerate limiting case, excluded by our axioms.

Which signature describes \emph{our} universe is an empirical question---either null directions exist or they do not---but much of the architecture of relativity does not depend on the answer. Independently of signature, every geometry the axioms allow carries a principle of relativity, with each observer's split of directions into a ``time'' axis and the ``space'' transverse to it; the relativity of simultaneity, since observers with different time axes have different orthogonal simultaneity hyperplanes; and reciprocal time dilation, length contraction, and a non-Galilean composition of velocities, all reflecting the foreshortening of one observer's axes when projected onto another's. What the indefinite signature adds, and the definite signature lacks, is the causal layer: a null cone, an invariant limiting speed, a causal ordering of events, and the hyperbolic---rather than trigonometric---form of the transformations. Interestingly, many of the features that make relativity counterintuitive are thus consequences of a nondegenerate quadratic form and a transitive symmetry group.

\medskip
\noindent\textbf{Relation to prior frameworks.}
The idea that symmetry alone pins down geometry has a long history, and it is worth situating our result among its predecessors. The classical \emph{space-problem} of Helmholtz and Lie \cite{Helmholtz1868,Lie1893} asks which geometries admit ``free mobility'' of rigid bodies---a transitive isometry group whose point-stabilizers act transitively on directions---and answers that, in the Riemannian setting, only the constant-curvature spaces (Euclidean, spherical, hyperbolic) qualify; the modern classification of two-point homogeneous spaces is due to Wang \cite{Wang1952} and Tits \cite{Tits1955}.

A complementary tradition classifies not the geometry but the kinematical group directly. The Cayley--Klein scheme organizes the planar kinematics---Galilean, Euclidean, and Minkowskian---into a single family \cite{Yaglom1979}, and Bacry and L\'evy-Leblond classified the possible kinematical groups in $3+1$ dimensions by Lie-algebra deformation and contraction \cite{BacryLevyLeblond1968}. These classifications begin from a distinguished \emph{time} direction and rotational isotropy of a \emph{space} presupposed distinct from \emph{time}, and enumerate the admissible group laws. Our route is the reverse---a single interval primitive with no \emph{a priori} time, from which the group is recovered---and our determinism axiom removes the degenerate members of their list (the Galilean and Carrollian kinematics, whose interval forms are degenerate), leaving exactly the nondegenerate signatures.

Our primitive is markedly weaker than that of the classical space-problem: rather than a positive-definite Riemannian metric we assume only a homogeneous interval function $D$, which need not be positive, need not satisfy a triangle inequality, and need not even be a norm. Our isotropy hypothesis (Axiom~\ref{axiom:isometry}) is a free-mobility condition in the same spirit, and mildly lighter: at a fixed reference direction we require only that each transverse direction be \emph{reversible} by an isometry, not that the stabilizer act transitively on all of them. That the quadratic forms we recover satisfy not merely this reversibility but the full free-mobility isotropy is well-established: by Witt's theorem on quadratic forms \cite{Witt1937}, the orthogonal group $O(S)$ acts transitively on the nonzero vectors of each fixed $S$-length, in every signature, so the quadratic geometries we recover do realize the axiom---confirming the hypothesis is non-vacuous. Even in the positive-definite case, that a sufficiently transitive isometry group forces an inner-product---hence quadratic---structure is itself known: for normed spaces it follows from point-homogeneity of the isometries \cite{DamaseKhare2023}, with a parallel rigidity for symmetric Minkowski norms \cite{XuMatveev2022}. Our contribution drops positivity and the norm axiom and carries the rigidity into the indefinite signatures, where these norm-based arguments do not apply.

Direction-dependent ``length'' functions of this kind are the subject of \emph{Finsler geometry} \cite{BaoChernShen2000}; our $D$ is a Finsler-type structure of homogeneity degree $p$, and the quadratic forms we recover are the flat, maximally symmetric members of that broader class.

Finally, our Theorem~\ref{thm:isometries-affine} (isometries are affine) is reminiscent of the Mazur--Ulam theorem \cite{MazurUlam1932}, that every surjective isometry of a normed space is affine; but that theorem requires a genuine norm and does not apply here, since $D$ is not a norm. We instead obtain affineness from the variational structure---isometries preserve geodesics, which are straight lines---via the fundamental theorem of affine geometry. The novelty of the present work is thus less the catalogue of geometries---which is long-established---than the route to it: a single interval primitive, with no \emph{a priori} distinction between time and space, from which the symmetry group, the quadratic structure, and the split between definite (Euclidean) and indefinite (Lorentzian and ultrahyperbolic) geometries all emerge rather than being assumed; indefinite and non-normed intervals are treated on the same footing as the Euclidean ones.

\smallskip
The kinematic structure of relativity, in the end, follows from symmetry and determinism.

\newpage
\appendix
\part*{Appendices}

\section{Proof of the Fundamental Theorem of Affine Geometry}
\label{app:affine-geometry}

We prove Lemma~\ref{lem:lines-to-affine}: a continuous bijection $T: \mathbb{R}^n \to \mathbb{R}^n$ mapping straight lines to straight lines must be affine.

\begin{proof}
\textbf{Step 1: Reduce to the origin-fixing case.}

Let $c = T(0)$ and $A(x) := T(x) - c$. Then $A(0) = 0$ and $A$ still maps lines to lines (the image of any line under $T$ is a line, so translating by $-c$ preserves this). It suffices to show $A$ is linear.

\textbf{Step 2: Lines through the origin map to lines through the origin.}

Consider $\ell(t) = tv$ for $v \neq 0$. Since $A$ maps this line to a line through $A(0) = 0$, we have $A(tv) = \varphi_v(t) \cdot A(v)$ for some continuous function $\varphi_v: \mathbb{R} \to \mathbb{R}$ with $\varphi_v(0) = 0$ and $\varphi_v(1) = 1$.

We claim $\varphi_v(t) = t$ for all $t$ and all $v$; this is established in Step 4.

\textbf{Step 3: $A$ maps planes to planes, and preserves parallelism.}

\emph{Plane-preservation.} Let $\Pi$ be the plane through three non-collinear points $p_0, p_1, p_2$. Since $A$ is a line-preserving bijection, its inverse also maps lines to lines; hence $A$ carries non-collinear points to non-collinear points, and $A(p_0), A(p_1), A(p_2)$ span a plane $\Pi'$. We claim $A(\Pi) = \Pi'$. Any $x \in \Pi$ lies on a line through some vertex $p_i$ that meets the opposite side-line at a point $q$ (choose $i$ so this line is not parallel to that side). Both $p_i$ and $q$ lie in $\Pi$, and $A(p_i), A(q) \in \Pi'$ (the side-line maps into $\Pi'$, being a line through two of its spanning points). Thus the line through $p_i$ and $q$ maps to a line meeting $\Pi'$ in two points, hence lying in $\Pi'$, so $A(x) \in \Pi'$. Therefore $A(\Pi) \subseteq \Pi'$; applying the same argument to $A^{-1}$ gives $A(\Pi) = \Pi'$.

\emph{Parallelism.} Suppose lines $\ell(t) = a + tv$ and $\ell'(t) = b + tv$ are parallel (same direction $v$, different basepoints $a \neq b$). They lie in a common plane $\Pi$, so their images lie in the plane $\Pi' = A(\Pi)$. Within $\Pi'$ two lines either meet or are parallel, so it suffices to rule out a common point. If $A(\ell)$ and $A(\ell')$ met, there would exist $t_1, t_2$ with $A(a + t_1 v) = A(b + t_2 v)$; injectivity gives $a + t_1 v = b + t_2 v$, so $(t_1 - t_2)v = b - a$. If $t_1 = t_2$ then $a = b$, a contradiction; if $t_1 \neq t_2$ then $b - a$ is parallel to $v$, making the lines identical, again contradicting $a \neq b$. Hence the coplanar lines $A(\ell), A(\ell')$ do not meet and are therefore parallel.

\textbf{Step 4: $\varphi_v(t) = t$, and $A$ is additive on independent pairs.}

Fix linearly independent $u, v$. Two families of parallel lines pass through the point $u + tv$:
\begin{itemize}
\item the line $\{u + sv : s \in \mathbb{R}\}$, parallel to $\{sv\}$; by Step~3 its image is the line through $A(u)$ in direction $A(v)$, so $A(u + tv) = A(u) + \psi(t)\,A(v)$ for a continuous $\psi$ with $\psi(0) = 0$;
\item the line $\{su + tv : s \in \mathbb{R}\}$, parallel to $\{su\}$; by Step~3 its image is the line through $A(tv) = \varphi_v(t)A(v)$ in direction $A(u)$, so $A(u + tv) = \varphi_v(t)\,A(v) + \rho(t)\,A(u)$ for some $\rho(t)$.
\end{itemize}
Because $u, v$ are independent, $A$ maps the plane $\langle u, v\rangle$ onto the plane $\langle A(u), A(v)\rangle$ (Step~3), so $A(u), A(v)$ are linearly independent. Equating the two expressions for $A(u + tv)$ in this basis gives $\rho(t) = 1$ and
\[
\psi(t) = \varphi_v(t), \qquad\text{i.e.}\qquad A(u + tv) = A(u) + \varphi_v(t)\,A(v).
\]
In particular $\psi$ does not depend on the basepoint $u$.

Now $\varphi_v$ is additive: applying the identity $A(u' + tv) = A(u') + \varphi_v(t)A(v)$ with basepoint $u' = u + sv$ (still independent of $v$) gives
\[
A\big(u + (s+t)v\big) = A(u + sv) + \varphi_v(t)\,A(v) = A(u) + \big(\varphi_v(s) + \varphi_v(t)\big)A(v),
\]
while applying it directly with basepoint $u$ gives $A(u + (s+t)v) = A(u) + \varphi_v(s+t)A(v)$. Comparing coefficients of $A(v)$ yields $\varphi_v(s+t) = \varphi_v(s) + \varphi_v(t)$. A continuous additive function is linear, and $\varphi_v(1) = 1$, so $\varphi_v(t) = t$ for all $t$. Hence
\[
A(tv) = t\,A(v) \quad\text{for all } v, t, \qquad\text{and}\qquad A(u + v) = A(u) + A(v) \text{ for independent } u, v.
\]

\textbf{Step 5: Linearity.}

Homogeneity $A(tv) = t\,A(v)$ holds for every $v$ and $t$ (Step~4). Additivity holds for independent pairs by Step~4, and for dependent pairs it follows from homogeneity: $A(u + cu) = A((1+c)u) = (1+c)A(u) = A(u) + A(cu)$. Thus $A$ is additive and homogeneous, hence linear, and $T(x) = A(x) + c$ is affine.
\end{proof}

\end{document}